\documentclass[twoside,11pt,final]{entics}
\pdfoutput=1

\usepackage[leftmargin=1em,rightmargin=1ex,vskip=1ex]{quoting}
\usepackage{lscape}
\usepackage{subcaption}
\usepackage{scalerel}
\usepackage{mathrsfs}
\usepackage{hyphenat}
\usepackage{amsmath}
\usepackage{hyperref}
\usepackage{cleveref}
\usepackage{enticsmacro}
\usepackage{doi}

\usepackage[utf8]{inputenc}
\makeatletter
\newcommand{\oset}[3][0ex]{%
  \mathrel{\mathop{#3}\limits^{
    \vbox to#1{\kern-2\ex@
    \hbox{$\scriptstyle#2$}\vss}}}}
\makeatother

\newcommand\scalemath[3]{\scalebox{#1}[#2]{\mbox{\ensuremath{\displaystyle #3}}}}

\newcommand{\leftarrowtip}{\ensuremath{\tikz\draw[line width=0.5pt,->] (10pt,0) -- (0,0);}}
\newcommand{\leftarrowtailnotip}{\ensuremath{\tikz\draw[line width=0.5pt,-<] (0,0) -- (10pt,0);}}

\newcommand{\unicodeStar}{\ensuremath{\star}}
\DeclareUnicodeCharacter{2605}{\unicodeStar}

\DeclareUnicodeCharacter{21D2}{\ensuremath{\Rightarrow}}
\DeclareUnicodeCharacter{2218}{\ensuremath{\circ}}
\DeclareUnicodeCharacter{2022}{\ensuremath{\bullet}}
\DeclareUnicodeCharacter{2219}{\ensuremath{\bullet}}
\DeclareUnicodeCharacter{2026}{\ensuremath{\dots}}
\DeclareUnicodeCharacter{2208}{\ensuremath{\in}}
\DeclareUnicodeCharacter{2192}{\ensuremath{\to}}
\DeclareUnicodeCharacter{2190}{\ensuremath{\leftarrowtip}}
\DeclareUnicodeCharacter{2919}{\ensuremath{\leftarrowtailnotip}}
\DeclareUnicodeCharacter{2270}{\ensuremath{\nleq}}
\DeclareUnicodeCharacter{2B47}{\ensuremath{\oset{\backsim}{\to}}}

\DeclareUnicodeCharacter{22A0}{\ensuremath{\boxtimes}}
\DeclareUnicodeCharacter{22B0}{\ensuremath{\boxtimes}}

\DeclareUnicodeCharacter{00D7}{\ensuremath{\times}}
\DeclareUnicodeCharacter{00B7}{\ensuremath{\cdot}}
\DeclareUnicodeCharacter{222B}{\ensuremath{\int}}
\DeclareUnicodeCharacter{22A4}{\ensuremath{\top}}
\DeclareUnicodeCharacter{22A5}{\ensuremath{\bot}}
\DeclareUnicodeCharacter{2264}{\ensuremath{\leq}}
\DeclareUnicodeCharacter{2225}{\ensuremath{\parallel}}

\newcommand{\unicodecolon}{\ensuremath{\colon}}
\DeclareUnicodeCharacter{FE55}{\unicodecolon}
\newcommand{\unicodeleftpar}{\ensuremath{\left(}}
\DeclareUnicodeCharacter{27EE}{\unicodeleftpar}
\newcommand{\unicoderightpar}{\ensuremath{\right)}}
\DeclareUnicodeCharacter{27EF}{\unicoderightpar}
\DeclareUnicodeCharacter{2260}{\neq}
\DeclareUnicodeCharacter{22A9}{\Vdash}
\DeclareUnicodeCharacter{2237}{\proportion}
\DeclareUnicodeCharacter{2124}{\mathbb{Z}}
\DeclareUnicodeCharacter{27E8}{\langle}
\DeclareUnicodeCharacter{27E9}{\rangle}
\DeclareUnicodeCharacter{21A6}{\mapsto}
\DeclareUnicodeCharacter{22A2}{\vdash}
\DeclareUnicodeCharacter{2090}{\ensuremath{{}_a}}
\DeclareUnicodeCharacter{A71B}{{}^\uparrow}
\DeclareUnicodeCharacter{A71C}{{}^\downarrow}
\DeclareUnicodeCharacter{27E6}{\llbracket}
\DeclareUnicodeCharacter{27E7}{\rrbracket}

\newcommand{\unicoderightcircle}{\ensuremath{\RIGHTcircle}}
\DeclareUnicodeCharacter{25D1}{\unicoderightcircle}
\newcommand{\unicodeleftcircle}{\ensuremath{\LEFTcircle}}
\DeclareUnicodeCharacter{25D0}{\unicodeleftcircle}
\DeclareUnicodeCharacter{229B}{\circledast}

\newcommand{\unicodebbA}{\ensuremath{\mathbb{A}}}
\DeclareUnicodeCharacter{1D538}{\unicodebbA}
\newcommand{\unicodebbB}{\ensuremath{\mathbb{B}}}
\DeclareUnicodeCharacter{1D539}{\unicodebbB}
\newcommand{\unicodebbC}{\ensuremath{\mathbb{C}}}
\DeclareUnicodeCharacter{2102}{\unicodebbC}
\DeclareUnicodeCharacter{1D53B}{\ensuremath{\mathbb{D}}}
\DeclareUnicodeCharacter{2113}{\ensuremath{\ell}}
\DeclareUnicodeCharacter{2115}{\ensuremath{\mathbb{N}}}
\DeclareUnicodeCharacter{211D}{\ensuremath{\mathbb{R}}}
\DeclareUnicodeCharacter{1D543}{\ensuremath{\mathbb{L}}}
\newcommand\UnicodeBlackboardP{\ensuremath{\mathbf{P}}} \DeclareUnicodeCharacter{2119}{\UnicodeBlackboardP}
\DeclareUnicodeCharacter{211A}{\ensuremath{\mathbb{Q}}}
\DeclareUnicodeCharacter{1D544}{\ensuremath{\mathbb{M}}}
\DeclareUnicodeCharacter{1D54C}{\ensuremath{\mathbb{U}}}
\DeclareUnicodeCharacter{1D54D}{\ensuremath{\mathbb{V}}}
\DeclareUnicodeCharacter{1D54E}{\ensuremath{\mathbb{W}}}
\DeclareUnicodeCharacter{1D546}{\ensuremath{\mathbb{O}}}
\DeclareUnicodeCharacter{1D540}{\ensuremath{\mathbb{I}}}
\DeclareUnicodeCharacter{1D54A}{\ensuremath{\mathbb{S}}}
\DeclareUnicodeCharacter{1D53C}{\ensuremath{\mathbb{E}}}

\newcommand{\unicodecalS}{\ensuremath{\mathcal{S}}}
\newcommand{\unicodecalT}{\ensuremath{\mathcal{T}}}
\newcommand{\unicodecalC}{\ensuremath{\mathcal{C}}}
\newcommand{\unicodecalD}{\ensuremath{\mathcal{D}}}
\newcommand{\unicodecalX}{\ensuremath{\mathcal{X}}}
\newcommand{\unicodecalN}{\ensuremath{\mathcal{N}}}
\newcommand{\unicodecalE}{\ensuremath{\mathcal{E}}}
\DeclareUnicodeCharacter{1D4D4}{\unicodecalE}
\DeclareUnicodeCharacter{1D4D2}{\unicodecalC}
\DeclareUnicodeCharacter{1D4D3}{\unicodecalD}
\DeclareUnicodeCharacter{1D4DE}{\mathcal{O}}
\DeclareUnicodeCharacter{1D4E2}{\unicodecalS}
\DeclareUnicodeCharacter{1D4E3}{\unicodecalT}
\DeclareUnicodeCharacter{1D4E7}{\unicodecalX}
\DeclareUnicodeCharacter{1D4D0}{\ensuremath{\mathcal{A}}}
\DeclareUnicodeCharacter{1D4D1}{\ensuremath{\mathcal{B}}}
\DeclareUnicodeCharacter{1D4D6}{\ensuremath{\mathcal{G}}}
\DeclareUnicodeCharacter{1D4D7}{\ensuremath{\mathcal{H}}}
\DeclareUnicodeCharacter{1D4DB}{\ensuremath{\mathcal{L}}}
\DeclareUnicodeCharacter{1D4DC}{\ensuremath{\mathcal{M}}}
\DeclareUnicodeCharacter{1D4DD}{\unicodecalN}
\DeclareUnicodeCharacter{1D4B1}{\ensuremath{\mathcal{V}}}
\DeclareUnicodeCharacter{1D4E5}{\ensuremath{\mathcal{V}}}
\DeclareUnicodeCharacter{1D4E6}{\ensuremath{\mathcal{W}}}
\DeclareUnicodeCharacter{1D4E4}{\ensuremath{\mathcal{U}}}

\DeclareUnicodeCharacter{03B1}{\alpha}
\DeclareUnicodeCharacter{03B2}{\beta}
\DeclareUnicodeCharacter{03BC}{\mu}
\DeclareUnicodeCharacter{03B4}{\delta}
\DeclareUnicodeCharacter{03B5}{\varepsilon}
\DeclareUnicodeCharacter{03B7}{\eta}
\DeclareUnicodeCharacter{03BB}{\lambda}
\DeclareUnicodeCharacter{03C1}{\rho}
\DeclareUnicodeCharacter{03C8}{\psi}
\DeclareUnicodeCharacter{03C4}{\tau}
\DeclareUnicodeCharacter{03A8}{\Psi}
\DeclareUnicodeCharacter{03C3}{\sigma}
\DeclareUnicodeCharacter{03C6}{\varphi}
\DeclareUnicodeCharacter{03A6}{\Phi}
\DeclareUnicodeCharacter{03A3}{\Sigma}
\DeclareUnicodeCharacter{03D5}{\phi}
\DeclareUnicodeCharacter{03B8}{\theta}
\DeclareUnicodeCharacter{03C0}{\ensuremath{\pi}}
\DeclareUnicodeCharacter{0393}{\Gamma}
\DeclareUnicodeCharacter{0394}{\Delta}
\DeclareUnicodeCharacter{03BA}{\kappa}
\DeclareUnicodeCharacter{03BD}{\nu}
\DeclareUnicodeCharacter{25A0}{\blacksquare}
\DeclareUnicodeCharacter{25AA}{\blacksquare}

\newcommand{\hirayo}{\scaleobj{0.9}{\text{\usefont{U}{min}{m}{n}\symbol{'210}}}}
\DeclareUnicodeCharacter{3088}{\hirayo}
\DeclareFontFamily{U}{min}{}
\DeclareFontShape{U}{min}{m}{n}{<-> udmj30}{}

\newcommand\UnicodeWhiteRightPointingSmallTriangle{\triangleright}
\DeclareUnicodeCharacter{25B9}{\mathbin{\UnicodeWhiteRightPointingSmallTriangle}}
\newcommand\UnicodeWhiteDownPointingSmallTriangle{\triangledown}
\DeclareUnicodeCharacter{25BF}{\mathbin{\UnicodeWhiteDownPointingSmallTriangle}}
\newcommand\UnicodeWhiteUpPointingSmallTriangle{\scalemath{1}{-1}{{}^{\triangledown}}}
\DeclareUnicodeCharacter{25B5}{\mathbin{\UnicodeWhiteUpPointingSmallTriangle}}

\DeclareUnicodeCharacter{2080}{\ensuremath{{}_0}}
\DeclareUnicodeCharacter{2081}{\ensuremath{{}_1}}
\DeclareUnicodeCharacter{2082}{\ensuremath{{}_2}}
\DeclareUnicodeCharacter{2083}{\ensuremath{{}_3}}

\DeclareUnicodeCharacter{1D62}{\ensuremath{{}_i}}
\DeclareUnicodeCharacter{2C7C}{\ensuremath{{}_j}}
\DeclareUnicodeCharacter{02B3}{\ensuremath{{}^r}}
\DeclareUnicodeCharacter{02E1}{\ensuremath{{}^\ell}}
\DeclareUnicodeCharacter{1D48}{\ensuremath{{}^d}}
\DeclareUnicodeCharacter{1D50}{\ensuremath{{}^m}}
\DeclareUnicodeCharacter{1D58}{\ensuremath{{}^u}}
\DeclareUnicodeCharacter{209A}{\ensuremath{{}_p}}
\DeclareUnicodeCharacter{2096}{\ensuremath{{}_k}}
\DeclareUnicodeCharacter{209C}{\ensuremath{{}_t}}

\DeclareUnicodeCharacter{2245}{\ensuremath{\cong}}
\DeclareUnicodeCharacter{2286}{\subseteq}

\DeclareUnicodeCharacter{22C5}{\cdot}
\DeclareUnicodeCharacter{25C3}{\ensuremath{\triangleleft}}
\DeclareUnicodeCharacter{25B9}{\ensuremath{\triangleright}}

\newcommand\smallmath[2]{#1{\raisebox{\dimexpr \fontdimen 22 \textfont 2
      - \fontdimen 22 \scriptscriptfont 2 \relax}{$\scriptscriptstyle #2$}}}
\newcommand\smalloplus{\smallmath\mathbin\oplus}

\DeclareUnicodeCharacter{2295}{\smalloplus}
\DeclareUnicodeCharacter{2297}{\otimes}
\DeclareUnicodeCharacter{214B}{\parr}
\DeclareUnicodeCharacter{2298}{\oslash}
\DeclareUnicodeCharacter{25C0}{\mathbin{\blacktriangleleft}}
\DeclareUnicodeCharacter{25C1}{\mathbin{\vartriangleleft}}
\DeclareUnicodeCharacter{22B3}{\mathbin{\triangleright}}

\DeclareUnicodeCharacter{FF5C}{\mid}
\DeclareUnicodeCharacter{227A}{\mathbin{\prec}}
\DeclareUnicodeCharacter{227B}{\mathbin{\succ}}
\DeclareUnicodeCharacter{22A3}{\mathbin{\dashv}}
\DeclareUnicodeCharacter{219D}{\ensuremath{\leadsto}}
\DeclareUnicodeCharacter{1361}{\colon}

\DeclareUnicodeCharacter{29D1}{\mathrel{\multimapdotbothB}}
\DeclareUnicodeCharacter{29D2}{\mathrel{\multimapdotbothA}}
\DeclareUnicodeCharacter{22C4}{\mathbin{\diamond}}
\DeclareUnicodeCharacter{226B}{\mathrel{\gg}}
\DeclareUnicodeCharacter{25A1}{\Box}
\DeclareUnicodeCharacter{266F}{\sharp}

\DeclareUnicodeCharacter{2190}{\gets}

\DeclareUnicodeCharacter{2099}{_n}
\DeclareUnicodeCharacter{2098}{_m}
\DeclareUnicodeCharacter{1D5AD}{\ensuremath{\mathsf{N}}}

\DeclareUnicodeCharacter{1D4DF}{\mathcal{P}}

\newcommand\mydots{\makebox[0.6em][c]{.\hfil.\hfil.}}
\DeclareUnicodeCharacter{2026}{\mydots}
\DeclareUnicodeCharacter{226B}{\gg}
\DeclareUnicodeCharacter{2248}{\UnicodeApprox}
\DeclareUnicodeCharacter{227C}{\preceq}
\newcommand{\UnicodeApprox}{\ensuremath{\approx}}

\usepackage{stmaryrd}
\newcommand{\unicodeRelationalComposition}{\fatsemi}
\DeclareUnicodeCharacter{2A3E}{\unicodeRelationalComposition}

\newcommand{\unicodeSqleq}{\mathbin{\ensuremath{\sqsubseteq}}}
\DeclareUnicodeCharacter{2291}{\unicodeSqleq}
\newcommand{\unicodeLhd}{\mathbin{\triangleleft}}\DeclareUnicodeCharacter{22B2}{\unicodeLhd} %
\newcommand{\sneakylink}[2]{\protect\hyperlink{#1}{#2}}
\makeatletter
\newcommand{\nicelinktarget}[1]{\Hy@raisedlink{\hypertarget{#1}{}}}
\makeatother

\newcommand\defining[1]{\nicelinktarget{#1}{}} %
\usepackage{xcolor}

\definecolor{nordred}{HTML}{bf616a}
\definecolor{bordeaux}{HTML}{821529}
\definecolor{bluelink}{HTML}{003399}
\definecolor{nordred}{HTML}{bf616a}
\definecolor{nordblue}{HTML}{81a1c1}
\definecolor{norddarkblue}{HTML}{5e81ac}
\definecolor{nordgreen}{HTML}{a3be8c}
\definecolor{nordnight}{HTML}{4c566a}

\AtEndPreamble{
  \RequirePackage{hyperref}
  \hypersetup{
    breaklinks = true,
    linktocpage,
    colorlinks = true,
  }
} %
\newcommand{\partialMarkovCategory}{\kl{partial Markov category}}
\newcommand{\partialMarkovCategories}{\kl{partial Markov categories}}

\newcommand{\tensor}[1][\otimes]{#1} %
\newcommand{\monoidalunit}[1][I]{#1} %
\newcommand{\dcomp}{\mathbin{\fatsemi}} %
\newcommand{\cat}[1]{\mathsf{#1}} %
\newcommand{\fun}[1]{\mathsf{\mathbf{#1}}} %
\newcommand{\Par}{\cat{Par}} %
\newcommand{\Rel}{\cat{Rel}} %

\newcommand{\inv}[1]{#1^{-1}} %

\newcommand\id{\mathrm{id}}

\newcommand{\swap}{\sigma} %

\renewcommand{\unicodeSqleq}{\mathbin{\kl[conditional inequality]{\ensuremath{\sqsubseteq}}}}
\renewcommand{\unicodeLhd}{\hyperlink{linkTriangleComposition}{\mathbin{\triangleleft}}}
\newcommand{\leqprob}{\unicodeSqleq} %
\newcommand{\conditionalInequality}{\hyperlink{linkSqleq}{conditional inequality}}
\newcommand{\condcomp}{\unicodeLhd} %

\usepackage[xcolor,no patch,hyperref,quotation,electronic]{knowledge}
\knowledgeconfigure{notion}
\knowledgestyle{notion}{color=nordnight}
\knowledgestyle{intro notion}{color=nordnight,italic}
\knowledgestyle{kl unknown}{}

\knowledge{ignore}
 | Markov categories
 | Discrete cartesian restriction category
 | Markov category
 | discrete Cartesian restriction categories
 | balanced copy-discard category
 | zero scalar
 | non-zero
 | cancellative
 | mean inequality

\knowledge{notion}
| deterministic 
| Deterministic
| deterministic map
| Deterministic map
| deterministic maps
| Deterministic maps

\knowledge{notion}
| total
| Total
| total map
| Total map

\knowledge{notion}
| restriction order
| restriction

\knowledge{notion}
| least conditional
| least conditionals

\newcommand{\normalisedby}{\mathbin{\kl[restriction]{\preceq}}}

\knowledge{notion}
| sharp witness property
| sharp witness
| sharp witnesses

\knowledge{notion}
| quasi-total
| quasi-total conditionals
| quasi-total conditional

\knowledge{notion}
| partial Markov category
| partial Markov categories
| Partial Markov category
| Partial Markov categories

\knowledge{notion}
| discrete partial Markov category
| discrete partial Markov categories
| Discrete partial Markov category
| Discrete partial Markov categories
\newcommand{\discretePartialMarkovCategory}{\kl{discrete partial Markov category}}

\knowledge{notion}
| cartesian restriction category
| cartesian restriction categories
| Cartesian restriction category
| Cartesian restriction categories

\knowledge{notion}
| cartesian bicategory of relations
| Cartesian bicategory of relations
| cartesian bicategories of relations
| Cartesian bicategories of relations
| cartesian bicategory
| Cartesian bicategory
| cartesian bicategories
| Cartesian bicategories

\knowledge{notion}
| balanced
| Balanced

\knowledge{notion}
| copy-discard category
| copy-discard categories
| Copy-discard category
| Copy-discard categories
| CD-category
| CD-categories
\newcommand{\copyDiscardCategory}{\kl{copy-discard category}}

\knowledge{notion}
| copy-discard-compare
| Copy-discard-compare
| copy-discard-compare structure
| Copy-discard-compare structure
| copy-discard-compare category
| copy-discard-compare categories
| Copy-discard-compare category
| Copy-discard-compare categories
| CDC-category
| CDC-categories

\knowledge{notion}
| effectful triple
| effectful triples
| Effectful triple
| Effectful triples

\knowledge{notion}
| conditional composition
| Conditional composition

\knowledge{notion}
| conditional inequality
| Conditional inequality
| conditional preorder
| Conditional preorder

\knowledge{notion}
| normalisation
| Normalisation
| normalisations
| Normalisations
| normalisation order
| Normalisation order

\knowledge{notion}
| means inequality
| Means inequality

\knowledge{notion}
| conditional
| conditionals
| Conditional
| Conditionals

\knowledge{notion}
| self-normalising
| Self-normalising

\knowledge{notion}
| bayesian inverse
| Bayesian inverse
| bayesian inverses
| Bayesian inverses

\knowledge{notion}
| Cauchy--Schwarz inequality
| Cauchy--Schwarz inequalities
| Cauchy-Schwarz inequality
| Cauchy-Schwarz inequalities

\knowledge{notion}
| cap
| caps
| Cap
| Caps

\usepackage{default}

\usepackage[export]{adjustbox}
  
\allowdisplaybreaks

\usepackage{graphicx}

\newcommand{\copier}{δ}
\newcommand{\comparator}{μ}
\newcommand{\comp}{\mathbin{\fatsemi}}

\newcommand{\leqprobw}[1]{\mathrel{\sqsubseteq^{#1}}}
\newcommand{\BorelSubStoch}{\cat{BorelStoch}_{\le 1}}
\newcommand{\FinSubStoch}{\cat{FinStoch}_{\le 1}}
\newcommand{\Real}{\mathbb{R}}

\def\conf{MFPS 2025} 	%
\volume{5}			%
\def\volu{5}			%
\def\papno{14}			%
\def\mydoi{16686}
\def\lastname{Di Lavore, Román, Sobociński, Széles} %
\def\runtitle{Local Order in Partial Markov Categories} %
\def\copynames{E. Di Lavore, M. Román, P. Sobociński, M. Széles}    %

\begin{document}
\begin{frontmatter}

\title{Order in Partial Markov Categories}
\author{Elena Di Lavore\thanksref{a}}
\author{Mario Román\thanksref{a}}
\author{Paweł Sobociński\thanksref{b}}
\author{Márk Széles\thanksref{c}}
\address[a]{University of Oxford, UK}
\address[b]{Tallinn University of Technology, Estonia}
\address[c]{Radboud University Nijmegen, The Netherlands}

\begin{abstract}
	Partial Markov categories are a recent framework for categorical probability
	theory that provide an abstract account of partial probabilistic computation
	with updating semantics. In this article, we discuss two order relations on
	the morphisms of a partial Markov category. In particular, we prove that
	every partial Markov category is canonically preorder-enriched, recovering
	several well-known order enrichments. We also demonstrate that the existence
	of codiagonal maps (comparators) is closely related to order properties of
	partial Markov categories. Finally, we introduce a synthetic version of the
	Cauchy--Schwarz inequality and, from it, we prove that updating 
	increases
	validity.
\end{abstract}
\begin{keyword}
  Markov categories, preorder enrichment, string diagrams,
  probabilistic inference, copy-discard categories
\end{keyword} 

\end{frontmatter}

\section{Introduction}
\label{sec:introduction}

\kl{Markov categories} \cite{fritz2020synthetic,cho2019disintegration} are a 
synthetic framework for probability theory. They allow one to reason 
about probabilistic processes using a few basic axioms that model key 
aspects of probabilistic computation. Morphisms in a \kl{Markov category} 
can be composed both in sequence and in parallel, via a symmetric 
monoidal structure. Every object is equipped with a 
commutative comonoid, which allows copying and discarding
information. An important axiom requires the existence of 
"conditionals", an abstraction of regular conditional probability 
distributions.
This simple setting has been successfully applied to various areas in 
probability theory \cite{Fritz2024,FritzK23,Jacobs2021}, semantics 
of probabilistic programming 
\cite{stein2021structural,stein2021compositional,AckermanFKKMRSY24}	, 
causal inference 
\cite{jacobs2021causal,Lorenz2023,Yin2022}
and information theory \cite{Perrone24}.

Morphisms in \kl{Markov categories} represent total probabilistic
computations. Totality is encoded by an axiom -- naturality of 
discarding -- declaring that
maps can be discarded without affecting the meaning of the process. 
However, many operations in probability theory are inherently 
partial. An important example is bayesian updating: if one tries to
learn from evidence that is
incompatible with the prior belief, updating is impossible. \kl{Markov
categories} lack the structure to express this partiality.

In order to express partiality, we need to drop naturality of discarding, and
thus work with copy-discard (CD) categories 
\cite{cho2019disintegration,corradini1999algebraic}.
Morphisms in such copy-discard categories correspond to partial probabilistic
computations. We would like to highlight two recent works following 
this approach. In \cite{Lorenz2023}, the authors start from a 
copy-discard category, and derive
conditional distributions from equality comparison and \kl{normalisation} of
subprobability kernels. While the theory is elegant, it is mostly suited for
discrete probability. In contrast, the framework of \kl{partial Markov categories}
\cite{2023partialmarkov} starts from a \kl{copy-discard category} with \kl{conditionals}. The
advantage of this approach is the abundance of models. Examples of \kl{partial
Markov categories} include Kleisli categories of `subprobability monads',
\kl{cartesian restriction categories}
\cite{cockett2002restriction,cockett2003restriction,cockett2007restriction}, and
cartesian bicategories of relations~\cite{CARBONI198711}.

The main contribution of this paper is to show that every \kl{partial Markov
category} is enriched (in the sense of Kelly~\cite{kelly1982basic}) in the
category of preorders and monotone maps. That is, every hom-set in a \kl{partial
Markov category} is a preorder, and the operations of the symmetric monoidal
category structure are monotone with respect to this preorder. The existence of
such an order-enrichment makes the setting particularly appealing for developing
the basics of inequational reasoning in synthetic probability theory. Crucial
properties of bayesian reasoning are inequational: e.g. that updating 
a prior with
evidence increases the likelihood of the evidence being true in the posterior.

\paragraph{Synopsis}
\noindent
Section~\ref{sec:preliminaries} recalls the basic theory of \kl{partial Markov
categories} and introduces some running examples.
Section~\ref{sec:conditional_inequality} defines the preorder enrichment on
\kl{partial Markov categories}. We then display several well-known 
order 
enrichments as
instances of our construction. Section \ref{sec:least_conditionals} studies
\kl{least conditionals} in \kl{partial Markov categories}. Our main results in
this section are about the relationship of \kl{least conditionals} and the
existence of comparators (Theorems~\ref{prop:comp_cap_least_disint}
and~\ref{prop:cap_from_conditional}). Finally,
Section~\ref{sec:validity_increase} is an application of the conditional
preorder. We derive a synthetic analogue of the validity-increase theorem under
suitable assumptions (Theorem~\ref{thm:validity_increase}). 

Throughout the paper we will use string-diagrammatic notation with the diagrams
written from left to right. For a detailed introduction to string-diagrammatic
calculi, see Selinger's survey~\cite{selinger2010survey}. Occasionally, we will
also apply equational reasoning. Thanks to the strictification theorem for
monoidal categories \cite{mac1998categories}, we may treat all monoidal
categories as if they were strict, without loss of generality.

\paragraph*{Related work}
\noindent
Inequalities play an important role in the theory of cartesian restriction categories~\cite{cockett2002restriction} and cartesian bicategories of relations~\cite{CARBONI198711}.
In fact, restriction categories have a canonical ordering on 
morphisms~\cite[\S 2.1.4]{cockett2002restriction} and can be seen as 
enriched categories~\cite{cockettGarner14}. Similarly, morphisms in 
bicategories of relations have a canonical ordering that forms an 
enrichment \cite{Nester2024}.
These canonical inequalities are instances of the ordering we 
introduce 
in this paper, see Propositions \ref{prop:cart_restr_cat_order_equiv} 
and \ref{prop:bicat_of_rel_order_equiv}. 

Order enrichments for categories of probabilistic maps have been considered
before, in the context of effectus theory~\cite{Cho2015,Cho2019a,Jacobs18}.
Recently, quasi-Markov categories were introduced as a framework to handle
partial operations in probability theory
\cite{fritz2025empiricalmeasuresstronglaws,mohammed2025}. Under reasonable
assumptions, such quasi-Markov categories are poset-enriched by the
\kl{restriction} order. Studying preorder-enrichment in quasi-Markov categories
is left for future work.

\section{Preliminaries}\label{sec:preliminaries}

\kl{Copy-discard categories} are theories of
processes -- symmetric monoidal categories -- where resources can be copied and
discarded. \kl{Partial Markov categories}~\cite{2023partialmarkov} are \kl{copy-discard categories} that
additionally have \kl{conditionals}: a factorization property often needed in
categorical probability.%

We go a step further and explicitly introduce \emph{copy-discard-compare (CDC)
categories} (Definition \ref{def:copyDiscardCompareCategory}): theories of
processes where we can assert the equality of two resources (``compare'' them).
\kl{Discrete partial Markov categories} are \kl{copy-discard-compare categories} with
\kl{conditionals} (Definition 
\ref{def:discretePartialMarkovCategory})~\cite{2023partialmarkov}.
Table~\ref{tab:cd_category_zoo} summarises the terminology related to variations of "copy-discard categories".

\begin{table}[h]
  \centering
  \begin{tabular}{|l||c|c|c|}
	\hline
	\textbf{Structure}                         
	& {\small\textbf{Maps are}}
	& {\small\textbf{Comparators/caps}} 
	& {\small\textbf{Conditionals}}
	\\
	\hline\hline
	\kl{Copy-discard category}
	&                    
	&            
	&              
	\\
	\hline
	\kl{Copy-discard-compare category}
	&                    
	& \multicolumn{1}{c|}{\checkmark}           
	&              
	\\
	\hline
	\kl{Cartesian restriction category}
	& "deterministic"                  
	&           
	&              
	\\
	\hline
	\kl{Discrete cartesian restriction category}
	& "deterministic"
	& \multicolumn{1}{c|}{\checkmark}           
	&              
	\\
	\hline
	Quasi-Markov category %
	& "quasi-total" 
	&       
	&              
	\\
	\hline
	\kl{Partial Markov category}                            
	&                    
	&        
	& \multicolumn{1}{c|}{\checkmark}            
	\\
	\hline
	\kl{Discrete partial Markov category}
	&                 
	& \multicolumn{1}{c|}{\checkmark}                   
	& \multicolumn{1}{c|}{\checkmark}                    
	\\
	\hline
	\kl{Markov category}                                    
	& "total"               
	&            
	&              
	\\
	\hline
	\kl{Markov category} with \kl{conditionals}                                    
	& "total"              
	&            
	& \multicolumn{1}{c|}{\checkmark}  
	\\
	\hline
	\end{tabular}
	\caption{Types of copy-discard categories. In each row, the 
	structure entails all properties 
	listed or marked with (\checkmark). }
	\label{tab:cd_category_zoo}
\end{table}

\begin{definition}{\textbf{(Copy-discard category)}}
  \label{def:copyDiscardCategory}%
  A \emph{""copy-discard category""}~\cite{corradini1999algebraic,cho2019disintegration} is a symmetric monoidal category $(\cat{C},\otimes, I)$ in which every object $X$ has a compatible commutative comonoid structure, consisting of a counit or \emph{discard}, $ε_X ፡ X → I$; and a comultiplication or \emph{copy}, $δ_X ፡ X → X ⊗ X$, satisfying the axioms below (also in Figure~\ref{fig:comm_comon_axioms}).
  \begin{enumerate}
    \item Comultiplication is associative, $δ_X ⨾ (δ_X ⊗ \id{}) = δ_X ⨾ (\id{} ⊗ δ_X).$
    \item Counit is neutral for comultiplication, $δ_X ⨾ (ε_X ⊗ \id{}) = \id{} = δ_X ⨾ (\id{} ⊗ ε_X).$
    \item Comultiplication is commutative, $δ_X ⨾ σ_{X,X} = δ_X$,
    where $σ_{X,X} : X \otimes X \to X \otimes X$ is the symmetry.
    \item Comultiplication is uniform, $δ_{X ⊗ Y} = (δ_{X} ⊗ δ_{Y}) ⨾ (\id{} ⊗ σ ⊗ \id{})$, and $δ_I = \id{}$.
    \item Counit is uniform, $ε_{X ⊗ Y} = ε_X ⊗ ε_Y$ and $ε_I = \id{}$.
  \end{enumerate}
  Notably, counits and comultiplications are not required to form natural
  transformations.
\end{definition}

\begin{figure}[h!]
  \[\cocommutativecomonoidAxiomsFig{}\]
  \caption{Structure and axioms of a commutative comonoid.}
  \label{fig:cocommutative-comonoid}
  \label{fig:comm_comon_axioms}
\end{figure}

The language of \kl{copy-discard categories} can express when 
morphisms are `deterministic', `total' and `total on their domain of 
definition' (\kl{quasi-total}).

\begin{definition}%
  \textbf{(Deterministic, total)} %
  \AP A morphism \(f ፡ X → Y\) in a \kl{copy-discard category} is
  \intro{deterministic} if it commutes with the copy, \(f \dcomp \cp_{Y} =
  \cp_{X} \dcomp (f \tensor f)\). It is \intro{total} if it commutes with the
  discard, \(f \dcomp \discard_{Y} = \discard_{X}\). 
\end{definition}

\begin{definition}%
  {\textbf{(Restriction \cite{2023partialmarkov,Lorenz2023})}} %
  \AP Let $\cat{C}$ be a \kl{copy-discard category}.
  \begin{enumerate}
	\item We say that a map $f : X → Y$ is a \intro[restriction order]{restriction} of
	$g : X →
	Y$ if $f = \copier_X \comp (g \otimes (f \comp \discard_{Y}))$. In this
	situation, we write $f \normalisedby g$. See Figure \ref{fig:normalisation}
	for an illustration.
	\item A map $f$ is called \intro{quasi-total} if $f \normalisedby 
	f$.
  \end{enumerate}
\end{definition}

\noindent
Clearly, all "deterministic" and "total" maps are \kl{quasi-total}.
The collection of \kl{quasi-total} maps is partially ordered by restriction~\cite[Lemma 98]{Lorenz2023}.

\begin{figure}[ht!]
		\[ f \normalisedby g \quad\text{if and only if}\quad
		\morphismFig{f} =
		\parinequalityconditionFig  \]
		\caption{The restriction order on quasi-total maps}
		\label{fig:normalisation}
\end{figure}

\begin{definition}
  \defining{linkTriangleComposition}%
  \defining{linkConditionalComposition}%
  The \intro{conditional composition} of two morphisms, $f ፡ X → A$ and
 $g ፡ A ⊗ X → B$, in a \copyDiscardCategory{} is the morphism defined by the equation in Figure~\ref{fig:conditionalComposition}.
 In other words, \kl{conditional composition} passes a copy of the input and the
 output of the first morphism to the second morphism.%
 \begin{figure}[ht!]
   \[(f ⊲ g) = δ_X ⨾ ((f ⨾ δ_A) ⊗ \id_X) ⨾ (\id_A ⊗ g) \qquad\qquad f \condcomp g = \conditionalcompositionDefFig\]
   \caption{Conditional composition.}%
   \label{fig:conditionalComposition}%
 \end{figure}
\end{definition}

\begin{definition}\textbf{(Conditionals)}
  \label{def:conditionals}%
  A \kl{copy-discard category} has \emph{\intro{conditionals}} if 
  every morphism \(f \colon X \to A \tensor B\) can be decomposed as 
  \(f = (f \dcomp (\id_{A} \tensor \discard_{B})) \condcomp c\), for 
  some \kl{quasi-total} \(c \colon A \tensor X \to B\). In this 
  situation we say that $c$ is a \kl{conditional} of $f$.
  For morphisms \(f \colon X \to A\) and \(g \colon A \to B\), a 
  \intro{bayesian inverse} of \(g\) with respect to \(f\) is a
  \kl{conditional} \(\bayesinv{g}{f} \colon B \tensor X \to A\) of \(f 
  \dcomp \cp_{A} \dcomp (g \tensor \id)\).
\end{definition}

Requiring "conditionals" to be "quasi-total" is natural but not strictly necessary for some of our developments.
The proof of Theorem \ref{prop:cap_from_conditional} relies on "conditionals" being "quasi-total".

\begin{definition}\textbf{(Partial Markov category \cite{2023partialmarkov})}\label{def:partialMarkovCategory}
  A \intro{partial Markov category} is a \kl{copy-discard category} with
  \kl{conditionals}.
\end{definition}

\begin{remark}\label{rem:self-normalising}
  In a partial Markov category, a morphism \(f\) is \kl{quasi-total}, \(f \normalisedby f\), if and only if its domain predicate, \(f \dcomp \discard\), is deterministic~\cite[Proposition 3.5]{2023partialmarkov}.
\end{remark}

We will later work with a particular subclass of \kl{copy-discard categories}: those that are \emph{\kl{balanced}}.
\kl{Balanced} \kl{copy-discard categories} cover a large class of examples, and they are particularly well-behaved with respect to the partial order we introduce in Definition~\ref{def:preorder-partial-markov}.
The name \emph{\kl{balanced}} is related to properties of idempotent morphisms, see~\cite[Theorem 4.1.8.]{Fritz2023}.

\begin{definition}%
  \label{def:balanced}%
  {\textbf{(Balanced copy-discard category~\cite{Fritz2023})}} %
  \AP A \kl{copy-discard category} $\cat{C}$ is \emph{\intro{balanced}} if the
  following implication holds for all appropriately typed morphisms.
  \[\includegraphics{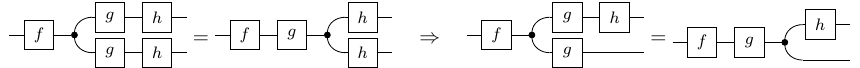}\]
\end{definition}

\noindent
Note that there are at least two other unrelated notions of 
\emph{balanced category} in the 
literature~\cite{selinger2010survey,Johnstone14}.
This should cause no confusion in the current context.

\begin{definition}{\textbf{(Copy-Discard-Compare category)}}
  \defining{linkStrictCopyDiscardCompareCategory}%
  \defining{linkCopyDiscardCompareCategory}%
  \label{def:copyDiscardCompareCategory}%
  A \intro{copy-discard-compare category} is a \kl{copy-discard category} in
  which every object, $X$, has a compatible partial Frobenius structure,
  consisting of an additional commutative multiplication, or 
  \emph{comparator}, $μ_X ፡ X ⊗ X → X$,
  satisfying the following axioms.
  \begin{enumerate}
    \item Multiplication is associative, $(μ_X ⊗ \id{}) ⨾ μ_X  = 
     (\id{} ⊗ μ_X) ⨾ μ_X $.
    \item Multiplication is commutative, $σ_{X,X} ⨾ μ_{X} = μ_{X}$.
    \item Multiplication is right inverse to comultiplication, $δ_X ⨾ μ_X = \id{}$.
    \item Multiplication satisfies the Frobenius rule, $(δ_X ⊗ \id{}) 
    ⨾ (\id{} ⊗ μ_X) = μ_X ⨾ δ_X = (\id{} ⊗ δ_X) ⨾ (μ_X ⊗ \id{}).$
    \item Multiplication is uniform, $μ_{X ⊗ Y} = (\id{} ⊗ σ_{Y,X} ⊗ \id{}) ⨾ (μ_{X} ⊗ μ_{Y})$, and $μ_I = \id{}$.
  \end{enumerate}
\end{definition}
\begin{figure}[h!]
  \[\comparatoraxiomsFig{}\]
  \caption{Structure and axioms of comparators.\label{fig:comparator-axioms}}
\end{figure}

\begin{definition}%
  {\textbf{(Discrete Partial Markov category~\cite{2023partialmarkov})}} %
  \label{def:discretePartialMarkovCategory}
  \AP A \intro{discrete partial Markov category} is a \kl{copy-discard-compare
  category} with \kl{conditionals}.
\end{definition}

We provide an alternative characterization of \kl{copy-discard-compare
categories} via \kl{cap} morphisms. This is the approach followed by 
Lorenz and Tull~\cite{Lorenz2023}, who mention that the
presentations via \kl{caps} and comparators are equivalent.
We provide a proof for completeness.

\begin{definition}{\textbf{(Caps~\cite[Definition 6]{Lorenz2023})}}
  \label{def:cap}
  A \kl{copy-discard category} is said to have \intro{caps} if it is 
  equipped with a family of
  morphisms $\cap_X : X ⊗ X → I$, satisfying the following axioms.
  \begin{enumerate}
	\item \label{def:cap:commutative}Caps are commutative, 
		$\sigma_{X,X} ⨾ \cap_X = \cap_X$.
		\item \label{def:cap:special}Caps interact with the comultiplication,
		$\copier_X \comp \cap_X = \discard_X$.
		\item \label{def:cap:frobenius}Caps satisfy the Frobenius rule, $(δ_X ⊗
		\id{}) ⨾ (\id{} ⊗ \cap_X) = (\id{} ⊗ δ_X) ⨾ (\cap_X ⊗ \id{})$.
		\item \label{def:cap:uniform}Caps are uniform, $\cap_{X ⊗ Y} = (\id{} ⊗
		σ_{Y,X} ⊗ \id{}) ⨾ (\cap_{X} ⊗ \cap_{Y})$, and $\cap_I = \id{}$.
	\end{enumerate}
\end{definition}

\begin{figure}[h!]
  \centering
  \includegraphics{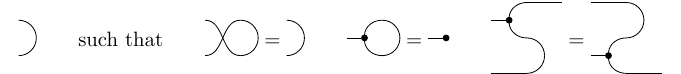}
  \caption{Structure and axioms of caps.\label{fig:cap-axioms}}
\end{figure}

\begin{proposition}%
  \label{prop:cap_iff_cdc}%
  A \kl{copy-discard category} is a \kl{copy-discard-compare category} if and
  only if it has \kl{caps}.
\end{proposition}
\begin{proof}
  Assume that $\cat{C}$ is a \kl{copy-discard-compare category}. We may define
  \kl{caps} via $\cap_X = \comparator_X ⨾ \discard_X$; the \kl{cap} axioms
  follow immediately.

  Conversely, assume that $\cat{C}$ has \kl{caps}. Define the comparator as one
  of the two sides of the Frobenius rule (\ref{def:cap:frobenius}),
  $\comparator_X = (δ_X ⊗ \id{}) ⨾ (\id{} ⊗ \cap_X) = (\id{} ⊗ δ_X) ⨾ (\cap_X ⊗
  \id{})$.
	
  First, we show commutativity using (\emph{a}) the definition of the compare
  map, (\emph{b}) naturality of symmetries, (\emph{c}) commutativity of the
  \kl{cap}, and (\emph{d}) commutativity of copy.
  \[\includegraphics{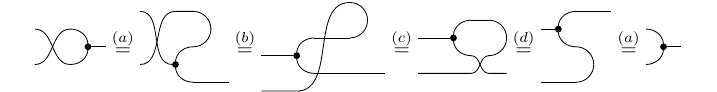}\]
  For associativity, we use (\emph{a}) the definition of the compare map, and
  (\emph{b}) associativity of copy.
  \[\comparefromcapsassociativeProofFig\]
  Compare is right inverse to copy by (\emph{a}) definition of the compare map,
  (\emph{b}) associativity of copy, (\emph{c}) axioms of caps, and (\emph{d})
  counitality of copy.
  \[\comparefromcapsspecialProofFig\]
  Finally, we show one of the Frobenius equations using (\emph{a}) 
  the definition of the compare map, and (\emph{b}) associativity of 
  copy. The other Frobenius equation can be proven in the same manner.
  \[\comparefromcapsfrobunniusProofFig\]
\end{proof}

\noindent
The rest of this section introduces some \kl{partial Markov categories} that we will use as
running examples.

\subsection{Example: subprobability kernels}

The category $\BorelSubStoch$ of standard Borel spaces and measurable
subprobability kernels is a \kl{discrete partial Markov
category}~\cite{2023partialmarkov}. The subcategory $\FinSubStoch
\hookrightarrow \BorelSubStoch$ of finite sets and substochastic matrices is
also a \discretePartialMarkovCategory{}. One can see that both $\BorelSubStoch$
and $\FinSubStoch$ are "balanced", by essentially the same argument as in
\cite[Proposition A.5.2.]{Fritz2023}.

\subsection{Example: cartesian restriction categories}
\intro{Cartesian restriction
categories}~\cite{cockett2002restriction,cockett2003restriction,cockett2007restriction}
are an abstraction of the category \(\Par\) of sets and partial functions. They
are \kl{copy-discard categories} where comultiplication is natural, i.e.\ all
morphisms are \kl{deterministic} but not necessarily total. 
\intro{Discrete cartesian restriction categories} are cartesian 
restriction
categories with a comparator.
We now show that \kl{cartesian restriction categories} have \kl{conditionals}.
That is, \kl{cartesian restriction categories} are \kl{partial Markov
categories} where all morphisms are \kl{deterministic}.

\begin{proposition}
  All \kl{cartesian restriction categories} are \kl{balanced} \kl{partial Markov
  categories}.
\end{proposition}
\begin{proof}
  Morphisms in \kl{cartesian restriction categories} are always
  \kl{deterministic}. We use this fact to show that there are \kl{conditionals}.
  \[\morphismOneTwoFig{f} = \parconditionalsProofFigOne =
  \conditionalsdeterministicFig\] 
  Again, because every morphism is \kl{deterministic}, these 
  \kl{conditionals} are \kl{quasi-total}, and the category is 
  \kl{balanced}.
\end{proof}

\subsection{Example: cartesian bicategories of relations}

\intro{Cartesian bicategories of relations} are an abstraction of the
category
\(\Rel\) of sets and relations. They are \kl{copy-discard-compare categories}
with extra structure and axioms that give them the expressive power of regular
logic (see Figure~\ref{fig:cartesian-bicategories}). They are also an example
class of \kl{partial Markov categories}.

\begin{figure}[h!]
  \cartesianbicategoriesAxiomFig{}
  \caption{Structure and axioms of cartesian bicategories of relations.\label{fig:cartesian-bicategories}}
\end{figure}

\begin{remark}\label{rem:cb-axioms}
  In every cartesian bicategory of relations, the copy and discard morphisms are not natural, but they are \emph{lax natural}: \(f \dcomp \cp \leq \cp \dcomp (f \tensor f)\) and \(f \dcomp \discard \leq \discard\) (\Cref{fig:original-cartesian-bicategories}, left).
  Moreover, the comonoid structure is left adjoint to the monoid structure: \(\id \leq \cp \dcomp \compare\), \(\compare \dcomp \cp \leq \id\), \(\id \leq \discard \dcomp \codiscard\) and \(\codiscard \dcomp \discard \leq \id\) (\Cref{fig:original-cartesian-bicategories}, middle and right).
  In fact, the original definition of cartesian bicategory of relations assumed a poset enrichment satisfying these conditions~\cite[Definitions 1.2 and 2.1]{CARBONI198711}.
  The definition we give here (\Cref{fig:cartesian-bicategories}) is equivalent to the original one~\cite[Lemma 4.1.6]{Nester2024} with the preorder enrichment defined in Equation~\ref{eq:inequality-rel}.
  \begin{figure}[h!]
    \centering
    \originalcartesianbicategoriesAxiomFig{}
    \caption{Comonoids in cartesian bicategories of relations are lax natural, and they are left adjoint to the monoids. \label{fig:original-cartesian-bicategories}}
  \end{figure}
\end{remark}

Idempotency of convolution (last equation in \Cref{fig:cartesian-bicategories}) implies that every morphism in a \kl{cartesian bicategory of relations} is \kl{quasi-total}.

\begin{lemma}\label{lemma:cb-self-normalising}
  In a cartesian bicategory of relations, every morphism is "quasi-total".
\end{lemma}
\begin{proof}
  We use the string diagrammatic syntax of \kl{cartesian bicategories of relations}.
  \[\relquasitotal\]
\end{proof}

Conditionals in \kl{cartesian bicategories of relations} are a consequence of the Frobenius structure and idempotency of convolution.

\begin{proposition}\label{lemma:rel-partial-markov}
  All \kl{cartesian bicategories of relations} are \kl{partial Markov categories}.
\end{proposition}
\begin{proof}
  By definition, \kl{cartesian bicategories of relations} are \kl{copy-discard categories}.
  We show that they also have conditionals using their string diagrammatic syntax~\cite{CARBONI198711}.
  The candidate conditional is in the dashed box below.
  This is a \kl{quasi-total} morphism because every morphism in a \kl{cartesian bicategory of relations} is \kl{quasi-total} (Lemma~\ref{lemma:cb-self-normalising}).
  We prove that the composite morphism on the left equals \(f\) by \emph{(i)} simplifying it to the composition on the right using the Frobenius equation and by \emph{(ii)} bounding the latter from above and from below with \(f\).
  \[\includegraphics{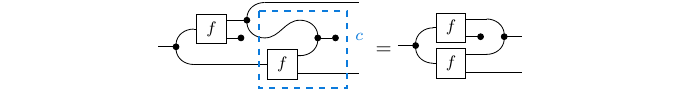}\]
  \kl{Cartesian bicategories of relations} are poset-enriched (see Equation~\ref{eq:inequality-rel}).
	We use this partial order to bound the morphism above from above with \(f\) using adjointness of
	the discard
	and the codiscard (\emph{i}), and lax naturality of the discard 
	morphism (\emph{ii}) (see \Cref{fig:original-cartesian-bicategories} and Remark~\ref{rem:cb-axioms}).
	\begin{equation*}
		\conditionalrelsimplifiedFig
		\overset{(i)}{\leq} \relconditionalProofFigOne
		\overset{(ii)}{\leq} \relconditionalProofFigTwo
		= \morphismOneTwoFig{f}
	\end{equation*}
	We bound the same morphism also from below with \(f\) using lax 
	naturality of
	the copy morphism (\emph{i}), and adjointness of the copy with 
	the cocopy (\emph{ii}) (see \Cref{fig:original-cartesian-bicategories} and Remark~\ref{rem:cb-axioms}).
	\begin{equation*}
		\conditionalrelsimplifiedFig
		\overset{(i)}{\geq} \relconditionalProofFigThree
		\overset{(ii)}{\geq} \relconditionalProofFigFour
		= \morphismOneTwoFig{f}
	\end{equation*}
\end{proof}

\noindent
\kl{Cartesian bicategories of relations} are not "balanced" in
general. One can
easily adapt a counterexample in the literature \cite[Example A.5.4.]{Fritz2023}
to get a counterexample in \(\Rel\).
\section{Conditional inequality}%
\label{sec:conditional_inequality}%

Let us endow \kl{partial Markov categories} with a 
preorder-enrichment. We prove
that the \kl{conditional inequality}
(Definition~\ref{def:preorder-partial-markov}) is an enrichment
(Theorem~\ref{thm:enrichment}) and derive basic inequalities.
	
\begin{definition}%
  \textbf{(Conditional inequality).} %
  \label{def:preorder-partial-markov}%
  In any \kl{copy-discard category}, \AP\intro{conditional inequality} relates
  two parallel morphisms, $f ⊑ g$ for $f ፡ X → Y$ and $g ፡ X → Y$, if and only
  if there exists some morphism $r ፡ Y ⊗ X → I$ such that $f = g ⊲ r$. In this
  case, we say that \(r\) \emph{witnesses} the inequality \(f ⊑ g\), and write
  $f \leqprobw{r} g$.
  \[\includegraphics*[valign=c]{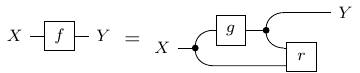}\]
\end{definition}

\begin{proposition}
  \kl{Conditional inequality} $(⊑)$ is a preorder.
\end{proposition}
\begin{proof}
  Reflexivity is witnessed by the discard morphism, \(f ⊲ \discard = f\), as on
  the left below. Let us prove transitivity. If we assume \(f \leqprobw{r} g\)
  and \(g \leqprobw{s} h\), then by associativity of conditional composition, $f
  = g ⊲ r = (h ⊲ s) ⊲ r = h ⊲ (s ⊲ r)$. Thus, $s ⊲ r$ witnesses $f ⊑ h$, as on
  the right below.
  \[\conditionalinequalityreflexiveProofFig{}
 \qquad \quad \conditionalinequalitytransitiveProofFig{}\] %
\end{proof}

\begin{theorem}%
  \label{thm:enrichment}%
  Any \kl{partial Markov category} is a preorder-enriched symmetric monoidal category.
\end{theorem}
\begin{proof}
  We show that the preorder $(⊑)$ on hom-sets is respected by
  composition and monoidal products.
  
  Assuming that $g \leqprobw{r} g'$, we can see that $f ⨾ g ⊑ f ⨾ g'$ 
  by
  $f ⨾ g = f ⨾ (g' ⊲ r) = (f ⨾ g') ⊲ (\bayesinv{g'}{f} ⊲ (r ⊗ ε)).$
  \[\includegraphics{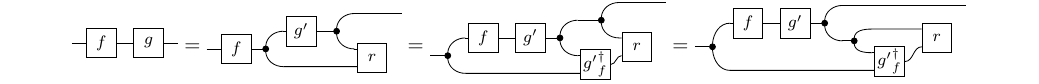}\]
  \indent Assuming that $g \leqprobw{r} g'$, we can see that $g ⨾ h ⊑ 
  g' ⨾ h$, 
  by
  $g ⨾ h = (g' ⊲ r) ⨾ h = (g' ⨾ h) ⊲ (\bayesinv{h}{g'} ⊲ (ε ⊗ r)).$
  \[\includegraphics{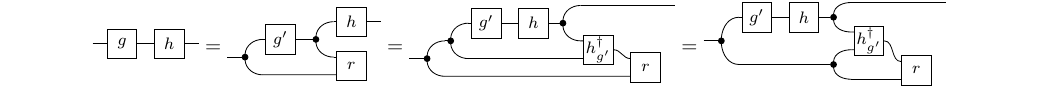}\]
  \indent Assuming that $g \leqprobw{r} g'$, we can see that $f ⊗ g ⊑ f ⊗ g'$, by
  $f ⊗ g = f ⊗ (g' ⊲ r) = (f ⊗ g') ⊲ (ε ⊗ r).$
  \[\includegraphics{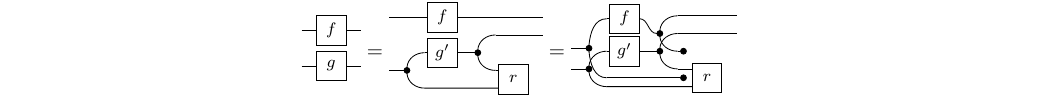}\]
  We proceed analogously to show that $f \leqprob f'$ implies $f ⊗ g ⊑ f' ⊗ g$.
  We showed that composition and tensors are monotone. Therefore, the
  \kl{conditional preorder} is an enrichment.
\end{proof}

\noindent
We collect a few elementary properties of the \kl{conditional preorder}.
\begin{proposition}
	Let $\cat{C}$ be a \discretePartialMarkovCategory{}.
	\begin{enumerate}
		\item For all $f : X \to Y$, $δ ⨾ (f ⊗ f) ⨾ μ ⊑ f$
		\item Multiplication and comultiplication form an adjoint 
		pair $\copier \dashv \comparator$ in the sense of 
		2-categories \cite{Lack2010}. 
		That is, $\id{} \leqprob \delta \comp 
		\mu$ and  $μ ⨾ δ ⊑ \id{}$.
	\end{enumerate}
\end{proposition}
\begin{proof}
  For point (i), observe that the equality below left follows from the partial
  Frobenius axioms.
  \[\includegraphics{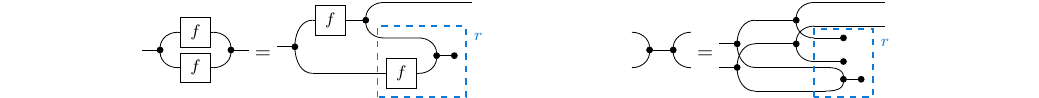}\]
  For point (ii), the first inequality holds by definition,
  while the second one can be derived from the partial Frobenius 
  axioms, as on the right above.
\end{proof}

\subsection{Subunital enrichments}

One may be interested in \kl{partial Markov categories} with 
additional logical structure, e.g., certain effectuses (in partial 
form) with copiers 
\cite{Cho2015}. In such categories, effects $X → I$ can 
be viewed as (fuzzy) predicates: for example, the discard map, $ε_X : 
X → I$, is interpreted as the `true' predicate. When equipping a 
\kl{partial Markov category} with a preorder enrichment, one would 
reasonably expect truth to be the largest predicate. 
We call a preorder enrichment \emph{subunital} if $p \le \discard_X$ for
all effects $p : X \to I$. 

\begin{proposition}\label{prop:subunitality}
	In any \partialMarkovCategory{}, both the \kl{conditional preorder} and the \kl{restriction order} are subunital.
\end{proposition}
\begin{proof}
	By counitality of copy.
	\[\discardtopelementProofFig\]
\end{proof}

Moreover, as the next proposition shows, the \kl{conditional 
preorder} is the least subunital preorder enrichment on a \kl{partial Markov category}. This property will be useful in
(\S\ref{subsec:condpreorder_examples}) where we compare 
the \kl{conditional preorder} to the canonical orderings in our running 
examples.

\begin{proposition}\label{prop:least_subunital_enrichment}
  Let $(\cat{C},\le)$ be a preorder-enriched \kl{partial Markov category} satisfying $p \le \discard_X$ for all effects $p \colon X \to I$.
  Then, for any two parallel morphisms, $f,g \colon X \to Y$, the following implication holds:
  \[ f \leqprob g \implies f \le g. \]
  An immediate consequence of this is that, if $(\le)$ is a 
  poset-enrichment, then the \kl{conditional preorder} is also a poset, 
  i.e. it is antisymmetric.
\end{proposition}
\begin{proof}
  Let $f \leqprob g$ be witnessed by $r \colon Y \otimes X \to I$. Then,
  \[\minimalpreorderProofFig{} \ .\]
\end{proof}

\subsection{Restriction and the conditional preorder}

In this subsection, we relate the "conditional 
preorder" to the restriction order.

\begin{lemma}
	\label{lem:cond_preord_implies_normalisedby}
	Let $\cat{C}$ be \kl{copy-discard category}, and let $f,g : X \to 
	Y$ be two parallel morphisms. Then, $f \normalisedby g$ 
	implies $f \leqprob g$. 
\end{lemma}
\begin{proof}
	The conditional inequality $f \leqprob g$ is witnessed by $f 
	\comp \discard_Y$.
\end{proof}

\noindent
The converse implication does not hold in general, see Example 
\ref{ex:orders_differ}. However, under 
certain assumptions, we can derive it for \kl{quasi-total} maps.

\begin{definition}{\textbf{(Sharp witness property)}} 
  A \kl{partial Markov category}, $\cat{C}$, satisfies the \intro{sharp witness property} if each inequality, $f \leqprob g$, between parallel
  "quasi-total" morphisms, $f,g : X \to Y$, is witnessed by a 
  deterministic effect $r : Y ⊗ X → I$.
\end{definition}

\begin{proposition}
	\label{prop:order_equivalence}	
	Let $\cat{C}$ be a \kl{balanced} \kl{copy-discard category} that satisfies the
	\kl{sharp witness property}. Let $f,g ፡ X → Y$ be parallel
	"quasi-total" maps. Then, $f \leqprob g$ if and only if $f
	\normalisedby g$.
\end{proposition}

\begin{proof}
  By Lemma~\ref{lem:cond_preord_implies_normalisedby}, \(f \normalisedby g\)
  implies \(f \leqprob g\). For the other direction, we use balance and \kl{sharp
  witnesses}. Suppose \(f \leqprob g\) with a \kl{sharp witness} \(r\). Since \(f\) is
  "quasi-total", we obtain:
  \[\orderequivalenceProofFigOne\]
  By balance, we obtain the equality below left. Then, by discarding 
  the bottom output, we obtain the equality \((a)\) below
  right, which shows that \(f \normalisedby g\).
  \[\orderequivalenceProofFigTwo\]
\end{proof}

\subsection{Examples}%
\label{subsec:condpreorder_examples}

We finish this section by instantiating the \kl{conditional preorder} in 
our running examples.

\paragraph{Finitary subdistributions}

Let $f, g ፡ X → \subdistr(Y)$ be two parallel substochastic channels in $\FinSubStoch$.
The inequality $f \leqprob g$ holds exactly when $f$ is pointwise dominated by $g$; that is, for all $y ∈ Y$ and $x ∈ X$, we have $f(y \given x) ≤ g(y \given x)$.
The category $\FinSubStoch$ also satisfies the \kl{sharp witness property}.
Consider two \kl{quasi-total} morphisms, $f,g
: X \to \subdistr(Y)$, satisfying $f \leqprob g$.
Then for all $x \in X$ either $f(y\given x) = 0$ for all $y \in Y$, or $f(y\given x) = g(y\given x)$ for all $y \in Y$.
We can then define a sharp predicate $r : Y \otimes X \to \subdistr(1)$ as below to witness the inequality \(f \leqprob g\).
\begin{align*}
	r(\ast \given y,x) = \begin{cases}
		0 &\text{ if } f(y' \given x) = 0 \text{ for all } y' \in Y 
		\\
		1 &\text{ otherwise }.
	\end{cases}
\end{align*}

\paragraph{Standard Borel spaces and subprobability kernels}
\label{ex:standardBorel}
One can define a pointwise order on parallel maps $f,g : (X,\sigma_X) 
\to (Y,\sigma_Y)$ in $\BorelSubStoch$ by setting $f \le g$ 
if and only if $f(T \given x) \le g(T \given x)$ for all points $x 
\in X$ and measurable sets $T \in \sigma_Y$. This pointwise order is 
a subunital enrichment \cite[Proposition 3.2.7. and 
Section 3.3.2.]{Cho2019a}. The next proposition shows that the 
\kl{conditional preorder} coincides with the pointwise order.

\begin{proposition}
	For any two parallel morphisms $f,g : (X,\sigma_X) \to 
	(Y,\sigma_Y)$ in $\BorelSubStoch$, $f \leqprob g$ if and only if 
	$f 
	\leq g$.
\end{proposition}
\begin{proof}
	\((\Rightarrow)\) The first implication is immediate from Proposition
	\ref{prop:least_subunital_enrichment}.
	
	\((\Leftarrow)\) Let $f ≤ g$ pointwise. The Radon-Nikodým Theorem for
	kernels (see \cite[Theorem 10]{Vakar2018}) states that there 
	exists
	an $r : (Y, σ_Y) × (X, σ_X) \to [0,\infty]$ such that $f(A\given x) =
	\int_{A} r(y,x) · g(dy|x)$ for all measurable sets $A ∈ σ_y$.
	It only remains to show that one can pick $r$ to be $[0,1]$-valued. Fix an
	$x ∈ X$, and let $B_x = \{ y ∈ Y \mid r(y,x) > 1 \}$. This $B_x$ is measurable. We show that $g(B_x|x) = 0$. By assumption,
	\[ ∫_{B_x} r(y,x) ·  g(dy \given x) 
	=  f(B_x\given x) 
	≤  g(B_x\given x) 
	=  ∫_{B_x} 1 · g(dy\given x).\] 
	By definition of $B_x$, the converse inequality also holds, so $\int_{B_x}
	r(y,x) \cdot g(dy\given x) = \int_{B_x} 1 \cdot g(dy\given x)$. Let $q_x(y)
	= (r(y,x) - 1) \cdot \chi(y \in B_x)$. Then $q \ge 0$, and also $\int_{B_x}
	q_x(y) \cdot g(dy \given x) = 0$. Therefore, $q_x$ is $g(-\given x)$
	-almost-surely zero, and thus $r(-,x) \le 1$, almost surely. We now set 
	\begin{align*}
		p(y,x) = \begin{cases}
			r(y,x), & \text{if } r(y,x) \le 1, \\
			0, & \text{otherwise}.
		\end{cases}
	\end{align*} 	
	This $p$ then satisfies $0 \le p \le 1$, and $f(A \given x) = 
	\int_{A} 
	p(y,x) \cdot 
	g(dy\given x)$ for all $A \in \sigma_Y$.
	
	Finally, we check that $p$ is jointly measurable. This is easy, since for
	any measurable $T ⊆ [0,1]$ we have
	\begin{align*}
	  p^{-1}(T) = \begin{cases}
		r^{-1}(T), & \text{if } 0 \notin T, \\
		r^{1}(T) \cup r^{-1}((1,\infty)), &\text{otherwise}.
	  \end{cases}
	\end{align*}
\end{proof}

From the previous proposition it is easy to see that $\BorelSubStoch$ 
also satisfies the \kl{sharp witness property}. The argument is very
similar to what we presented for $\FinSubStoch$.

\paragraph{Cartesian restriction categories}

We now instantiate Definition~\ref{def:preorder-partial-markov} in \kl{cartesian
restriction categories}. Such categories have a canonical poset
enrichment~\cite[\S 2.1.4]{cockett2002restriction}, where $f ≤ g$ if and only if
$f \normalisedby g$. This means that $f$ is a restriction of $g$ on a smaller
domain. Since all morphisms are deterministic, "cartesian restriction
categories" satisfy the "sharp witness property".

\begin{proposition}%
  \label{prop:cart_restr_cat_order_equiv}%
  In a \kl{cartesian restriction category}, the \kl{conditional 
  preorder} and the \kl{restriction order} coincide: \(f ⊑ g\) if and 
  only if \(f 
  \normalisedby g\).
\end{proposition}
\begin{proof}
  Since \kl{cartesian restriction categories} are \kl{balanced} and satisfy the
  \kl{sharp witness property}, this is immediate from Proposition
  \ref{prop:order_equivalence}.
\end{proof}

\paragraph{Cartesian bicategories of relations}
Let us instantiate Definition~\ref{def:preorder-partial-markov} in \kl{cartesian
bicategories of relations}. These also have a canonical poset enrichment where
\(f \leq g\) means that \(f\) is a `subrelation' of \(g\), in the following
sense:
\begin{equation}\label{eq:inequality-rel}
  f \leq g \quad\text{if and only if}\quad \morphismFig{f} =
  \convolutionFig{g}{f} \ .
\end{equation}
This presentation of the poset enrichment is equivalent to the usual
one~\cite[Lemma 4.1.5]{Nester2024}. The next proposition relates the two
orderings.

\begin{proposition}
	\label{prop:bicat_of_rel_order_equiv}
	In a \kl{cartesian bicategory of relations}, for two parallel
	morphisms \(f, g \colon X \to A\), \(f \leqprob g\) if and only 
	if \(f \leq g\).
\end{proposition}
\begin{proof}
	\((\Rightarrow)\) We apply 
	Proposition~\ref{prop:least_subunital_enrichment}
	because \((\leq)\) is also subunital.\\
	\((\Leftarrow)\) If \(f \leq g\) (\emph{i}), then we may take 
	\(r = (\id{} \tensor f) \dcomp \compare \dcomp \discard\) and 
	obtain that \(f \leqprob g\) by the Frobenius axioms (\emph{ii}).
	\[\morphismFig{f} \overset{(i)}{=} \convolutionFig{g}{f} 
	\overset{(ii)}{=} \relinequalityProofFigOne \]
\end{proof}

Every morphism in a \kl{cartesian bicategory of relations} is
"quasi-total". However, since these categories are not 
"balanced", the "restriction" order may differ from the "conditional 
preorder". The following example demonstrates that this is indeed the 
case.

\begin{example}
	\label{ex:orders_differ}
	Let $R,S \subseteq \mathbb{N} \times \mathbb{N}$ be morphisms in 
	\(\Rel\), given by $R = \{ (1,2) \}$ and $S = \{ (1,2),(1,3) \}$. 
	Then $R \le S$, but $R \not\preceq S$, since $\delta 
	\comp (R \otimes (S \comp \discard)) = R \ne S$.
\end{example}

\paragraph{Quasi-Markov categories}
Quasi-Markov categories are \kl{copy-discard categories} where all morphisms are \kl{quasi-total}~\cite{fritz2025empiricalmeasuresstronglaws}.
These are poset-enriched with the \kl{restriction} order~\cite[Proposition~2.16]{fritz2025empiricalmeasuresstronglaws}.

\begin{proposition}
  In a quasi-Markov category with conditionals, for two parallel morphisms \(f,g \colon X \to A\), \(f \leqprob g\) if and only if \(f \normalisedby g\).
\end{proposition}
\begin{proof}
  A quasi-Markov category with conditionals is, in particular, a \kl{partial Markov category}.
  By Lemma~\ref{lem:cond_preord_implies_normalisedby}, if \(f \normalisedby g\), then \(f \leqprob g\).
  By Proposition~\ref{prop:subunitality}, since every map is 
  quasi-total, the \kl{restriction 
  preorder} is subunital order enrichment.
  We can, then, apply Proposition~\ref{prop:least_subunital_enrichment} to obtain the other implication.
\end{proof}

\section{Least conditionals}\label{sec:least_conditionals}

Asking \partialMarkovCategories{} to have unique \kl{conditionals} is too 
strong. In any \partialMarkovCategory{} with unique \kl{conditionals}, \(f 
\comp \discard_{Y} = g \comp 
\discard_{Y}\) implies $f = g$ for all parallel maps $f,g \colon X 
\to Y$ \cite{2023partialmarkov}. As a special case, Markov 
categories with unique conditionals are posetal, see 
\cite[Remark 11.16]{fritz2020synthetic}. 

However, once we have an order in \partialMarkovCategories{}, there might be a
canonical choice of conditional, namely the smallest among all \kl{conditionals}. In
this section, we consider \kl{least conditionals} with respect to two orders on
"quasi-total" maps: the \kl{conditional preorder} $(\leqprob)$ and the
"restriction" order $(\normalisedby)$.

\begin{definition}%
  {\textbf{(Least conditionals)}} %
  Let $f: X → Y ⊗ Z$ be a morphism in a (locally small) \kl{partial Markov category}.
  Write $C_f = \{c ፡ Z ⊗ X → Y \mid f = (f ⨾ (\id{} ⊗ \discard)) \condcomp c \land c \normalisedby c\}$ for the set of "quasi-total" \kl{conditionals} of $f$.
  Now, we say that a \kl{conditional}, $g ∈ C_f$, is
  \begin{enumerate}
    \item a $\leqprob$-\intro{least conditional} of $f$ when it is a least
    element of
    $C_f$, ordered by the \kl{conditional preorder}, $(\leqprob)$;
    \item or, a $\normalisedby$-least conditional of $f$ if it is the least
    element of $C_f$ ordered by the \kl{restriction order} $(\normalisedby)$.
  \end{enumerate}
  By Lemma \ref{lem:cond_preord_implies_normalisedby}, every ($\normalisedby$)-least \kl{conditional} is a ($\leqprob$)-least
  \kl{conditional}.
\end{definition}

\subsection{Least conditionals and copy-discard-compare categories}

We demonstrate how \kl{least conditionals} are related to the existence of
comparators. We start by characterizing \kl{conditionals} of copy and 
identity maps in a \kl{copy-discard category}.

\begin{lemma}\label{lem:copy_id_disint}
  Let $\cat{C}$ be a \kl{copy-discard category}.
  \begin{enumerate}
    \item\label{lem:copy_id_disint_copy} A morphism $f \colon X ⊗ X → X$ is a
      \kl{conditional} of the copy $\cp_X ፡ X → X ⊗ X$ if and only if $\cp_X ⨾ f
      = \id_X$.
    \item\label{lem:copy_id_disint_id} A morphism $g ፡ X ⊗ X → I$
      is a \kl{conditional} of the identity $\id_X ፡ X → X$ if and only
      if $\cp_X ⨾ g = ε_X$.
  \end{enumerate}
  Note that Condition~(\ref{lem:copy_id_disint_copy}) holds for the multiplication ($\comparator_X$) by definition of \kl{copy-discard-compare category}.
  Similarly, the definition of the \kl{cap} $(\cap_X)$ contains Condition~(\ref{lem:copy_id_disint_id}).
\end{lemma}
\begin{proof}
  For point (i), suppose that $\cp_{X} ⨾ f = \id_X$. Then, the defining equation
  of conditionals can be derived as below left. Conversely, if $f$ is a
  conditional of $\copier_X$, then the middle equality on the right holds by
  definition. Is then easy to see that $\copier_X ⨾ f = \id_X$.	 
  \[\includegraphics{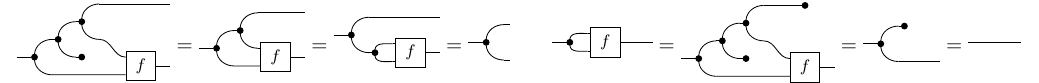}\]
  For point (ii), suppose that $\cp_{X} \dcomp g =\discard_X$. Then, the
  defining equation of conditionals can be derived as below left.
  Conversely, discarding the output on both sides of the equality on the left
  below shows $\cp_{X} \dcomp g =\discard_X$, as below right.
  \[\includegraphics{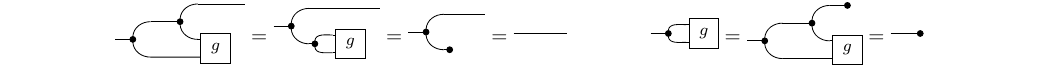}\]
\end{proof}

\begin{lemma}%
  \label{lem:id_disint_from_copy_disint} %
  Let $\cat{C}$ be a \kl{copy-discard category}. If $f ፡ X ⊗ X → X$ is a
  \kl{conditional} of the copy $\cp_X ፡ X → X ⊗ X$, then $f ⨾ ε_X ፡ X ⊗ X → I$
  is a \kl{conditional} of the identity map $\id_X ፡ X → X$. Moreover, if $f$ is
  the $\normalisedby$-\kl{least conditional} of $\copier_X$ then $f ⨾ ε_X$ is the
  $\normalisedby$-\kl{least conditional} of $\id_X$.
\end{lemma}
\begin{proof}
  By Lemma~\ref{lem:copy_id_disint}, $f ⨾ ε_X$ is a conditional of the identity.
  Now assume that $f$ is the $\normalisedby$-\kl{least conditional} of the copy map,
  and let $g ፡ X ⊗ X → I$ be a conditional of the identity. We have to show that
  $f ⨾ ε_X \normalisedby g$. Let $t = (\id_X ⊗ \cp_X) ⨾ (g ⊗ \id_X)$. Step
  $(a)$ below left uses that $g$ is a \kl{conditional} of the identity to show that
  this $t$ is a conditional of the copy map \(\cp_{X}\). Since $f$ is the
  $\normalisedby$-\kl{least conditional} of the copy map, $\cp_X$, then $f
  \normalisedby t$ (below right).
  \[\includegraphics{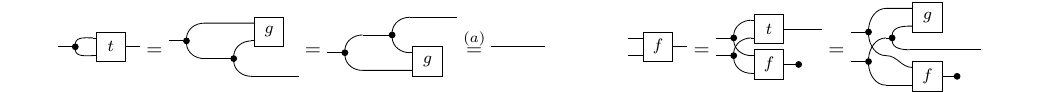}\]
  Discarding both sides of the above equation yields $f ⨾ ε_X  \normalisedby g$.
\end{proof}

We are ready to show that, if they exist, comparators and \kl{caps} are least
conditionals.

\begin{theorem}%
  \label{prop:comp_cap_least_disint} %
  Let $\cat{C}$ be a \kl{copy-discard-compare category}. The multiplication $μ_X
  ፡ X ⊗ X → X$ is the $\normalisedby$-\kl{least conditional} of the copy map $\cp_X ፡
  X → X ⊗ X$. The \kl{cap}, $\cap_X ፡ X ⊗ X → I$, is the \kl{least conditional} of
  the identity $\id_X ፡ X → X$ with respect to the \kl{restriction 
  order} $(\normalisedby)$. By Lemma 
  \ref{lem:cond_preord_implies_normalisedby},
  $\comparator_X$ and $\cap_X$ are also $\leqprob$-\kl{least conditionals}.
\end{theorem}
\begin{proof}
  By Lemma~\ref{lem:copy_id_disint}, $\comparator_X$ and $\cap_X$ are conditionals of the copy and identity maps, respectively.
  We show that $\comparator_X$ is the $\normalisedby$-\kl{least conditional} of $\cp_X$.
  Let $f \colon X \tensor X \to X$ be any conditional of the copy map.
  Using \emph{(i)} the Frobenius equations and \emph{(ii)} Lemma~\ref{lem:copy_id_disint}, we show that \(f \normalisedby \comparator_{X}\).
  \[\comparatorleastconditionalProofFig{}\]
  By Lemma~\ref{lem:id_disint_from_copy_disint}, the cap is the 
  $\normalisedby$-\kl{least conditional} of the identity.
\end{proof}

Since $\normalisedby$-\kl{least conditionals} are unique, the following is an
immediate consequence of the previous proposition.

\begin{corollary}
  \label{cor:cdc_unique}
  A \kl{copy-discard category} can admit at most one \kl{copy-discard-compare}
  structure.
\end{corollary}

Next, we prove a converse to Theorem~\ref{prop:comp_cap_least_disint}. We
demonstrate that, under certain conditions, \kl{least conditionals} of copier maps
give rise to a \kl{copy-discard-compare structure}. As it turns out, one needs
to consider the \kl{restriction order} $\normalisedby$ to prove the 
following
proposition.

\begin{theorem}\label{prop:cap_from_conditional}
  Let $\cat{C}$ be a \kl{copy-discard category}, and assume the following.
  \begin{enumerate}
    \item All copy maps $\cp_X ፡ X → X ⊗ X$ have a $\normalisedby$-least
    conditional, $μ_X ፡ X ⊗ X → X$. By
    Lemma~\ref{lem:id_disint_from_copy_disint}, the identity maps then also have
    $\normalisedby$-\kl{least conditionals} $\cap_X = \comparator_X ⨾
    \discard_X ፡ X ⊗ X → I$.
    \item The family of conditionals $\cap_X$ satisfy Condition~(\ref{def:cap:uniform}) of
    Definition~\ref{def:cap}.
  \end{enumerate}
  In this situation, the conditionals $\cap_X$ satisfy the axioms of 
  "caps". By Proposition~\ref{prop:cap_iff_cdc}, $\cat{C}$ is thus a
  \kl{copy-discard-compare category}.
\end{theorem}
\begin{proof}
  We show that the equations
  (\ref{def:cap:commutative})-(\ref{def:cap:frobenius}) of
  Definition~\ref{def:cap} are satisfied. The equality (\ref{def:cap:special})
  $\cp_X \comp \cap_X =\discard_X$ holds by Lemma~\ref{lem:copy_id_disint}. We
  show commutativity. Using commutativity of the copy map and the equality in Condition~\((i)\),
  we obtain that $\cp_X ⨾ \swap_{X,X} ⨾ \cap_X = \discard_X$ also holds (below
  left). By Lemma~\ref{lem:copy_id_disint}, this implies that $\swap_{X,X} ⨾
  \cap_X$ is also a conditional of $\id_X$. Therefore, $\cap_X \normalisedby
  \swap_{X,X} ⨾ \cap_X$. We use this fact in steps $(a)$ and $(c)$ below right.
  Step $(b)$ applies naturality of the symmetries.
  \[\includegraphics{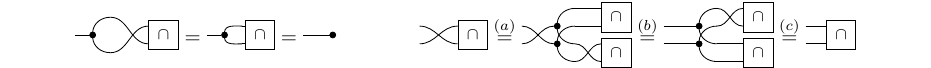}\]
  \noindent
  At last, we need to show equation (\ref{def:cap:frobenius}), whose two sides
  we name \(t\) and \(u\).
  \[\capfrobeniusequationProofFig{}\]
  Using coassociativity of copy, equality \((i)\) and counitality of copy, one
  can easily show that \(\cp ⨾ t = \id = \cp ⨾ u\). By
  Lemma~\ref{lem:copy_id_disint}, we obtain that both $t$ and $u$ are
  \kl{conditionals} of the copy $\cp_X ፡ X → X ⊗ X$. We now prove that $t =
  \comparator_X = u$. Steps marked with $(a)$ use minimality of $\comparator_X$.
  Equations $(b)$ use associativity and commutativity of the copier. Steps $(c)$
  use that $\cap_X$ is \kl{quasi-total}. 
  \[\capfrobeniusequationProofFigOne{}\]
  \[\capfrobeniusequationProofFigTwo{}\]
  This shows that \(\cat{C}\) is a \kl{copy-discard category} with \kl{caps},
  which, by Proposition~\ref{prop:cap_iff_cdc}, is also a
  \kl{copy-discard-compare} category.
\end{proof}

The following is an immediate consequence of the previous theorem and 
Proposition \ref{prop:order_equivalence}.

\begin{corollary}
	Let $\cat{C}$ be a \kl{balanced} \kl{partial Markov category} 
	that satisfies the 
	\kl{sharp witness property}. Assume 
	that the copier and identity maps have $\leqprob$-least 
	conditionals 
	$\comparator_X$ and $\cap_X$, and that the $\cap_X$ maps satisfy 
	Condition~(\ref{def:cap:uniform}) of
	Definition~\ref{def:cap}. Then $\cat{C}$ is a \kl{discrete 
	partial Markov category}.
\end{corollary}

\subsection{Examples}

\begin{proposition}
	The category \(\FinSubStoch\) is \kl{balanced} and satisfies the \kl{sharp witness property}. Therefore, $\normalisedby$-least
	conditionals 
	are the same as $\leqprob$-\kl{least conditionals}.
	For a morphism $f \colon X \to A \otimes B$, indicate the total probability mass of the subdistribution $f(-\given x)$ with $\|f(-\given x)\| = \sum_{b'\in B} f(a,b' \given x)$.
  The least conditional $c_0 : A \otimes X \to B$ of $f$ is given by
	\begin{align*}
		c_{0}(b \given a,x)& = \begin{cases} \frac{f(a,b \given 
				x)}{\sum_{b' \in B} f(a,b' \given x)} & \text{if 
			}\|f(-\given x)\| \neq 0,\\ 0 & \text{if 
			}\|f(-\given x)\| = 0. \end{cases}
	\end{align*}
\end{proposition}

The \kl{least conditional} is the one that contains no `junk information': for all elements $x \in X$ where \(f\) has positive probability of success, $\|f(- \given x)\| \ne 0$, the value $c(b \given a,x)$ is forced to be equal to $c_0(b \given a,x)$, for any \kl{conditional} $c : A ⊗ X → B$ of $f$;
for all elements \(x \in X\) where \(f\) certainly fails, $\|f(-\given x)\| = 0$, then $c(- \given a,x)$ could be any subdistribution.

The category $\BorelSubStoch$ does not have all \kl{least conditionals}.
We demonstrate this via the following example, adapted from 
\cite[Remark 5.2.10]{Cho2019a}.
\begin{example}
	Let $\mu$ be the Lebesgue-measure on the unit interval $[0,1]$, 
	viewed as a morphism $\mu : 1 \to [0,1] \otimes 1$ in 
	$\BorelSubStoch$. "Quasi-total" "conditionals" $c : [0,1] 
	\otimes 1 \to 1$ of $\mu$ correspond to Lebesgue-measurable 
	subsets $C \subseteq [0,1]$ such that for all measurable $B 
	\subseteq [0,1]$, $\mu(B) = \mu(B \cap C)$. The "restriction 
	order" $(\normalisedby)$ and the pointwise order $(\leqprob)$ 
	coincide on these conditionals, and translate to subset inclusion 
	on measurable subsets.
	For an $\alpha \in [0,1]$, the set $C_{\alpha} = [0,1] \setminus 
	\{\alpha\}$ satisfies $\mu(B) = \mu(B \cap C_{\alpha})$ for 
	all 
	measurable $B$. Therefore, a \kl{least conditional} $C_0$ must satisfy
	$C_0 \subseteq B_\alpha$ for all $\alpha \in [0,1]$. But such a 
	$C_0$ can only be the empty set $\emptyset$, which does not 
	satisfy $\mu(B 
	\cap \emptyset) = \mu(B)$ for all measurable $B$. 
	Therefore, a \kl{least conditional} of $\mu$ cannot exist.
\end{example}	
 
By Proposition \ref{prop:cart_restr_cat_order_equiv}, 
$\normalisedby$-\kl{least conditionals} are the same as $\leqprob$-least
conditionals in a discrete cartesian restriction category.
 
\begin{proposition}

  Any discrete cartesian restriction category has \kl{least conditionals}. Given a
  map $f : X → A ⊗ B$, the \kl{least conditional} of $f$ is given by \(c_{0} = (\id{}
  ⊗ f) ⨾ ((μ ⨾ ε) ⊗ \id{})\).
\end{proposition}
\begin{proof}
  We check that this is indeed a "quasi-total" "conditional". It 
  is "quasi-total" because every morphism is deterministic 
  (\emph{i}). We check
  the "conditional" equation using the partial Frobenius axioms 
  (\emph{ii}).
	\[\morphismOneTwoFig{f} \overset{(ii)}{=} 
	\relconditionalProofFigThree \overset{(i)}{=} 
	\conditionalrelsimplifiedFig \overset{(ii)}{=} 
	\conditionalbendingwireFig{f}{c_{0}}\]
  Suppose there is another "conditional", $c ፡ A ⊗ X → B$, of $f$ 
  (\emph{iii}).
	\[\parminimalconditionalProofFigOne \overset{(ii)}{=} 
	\parminimalconditionalProofFigTwo \overset{(iii)}{=} 
	\minimalconditionalbendwire = \morphismTwoOneFig{c_{0}}\]
  Then, \(c_{0} ≤ c\).
\end{proof}

To provide intuition, we explicitly describe \kl{least conditionals} in $\Par$.
The situation is analogous to $\FinSubStoch$.
Let $f \colon X \to A \otimes B$ be a partial function.
Suppose that \(f(x) = (a,b)\).
Then, every conditional \(c \colon A \tensor X \to B\) has to satisfy $c(a,x) = b$.
However, for all \(a'\) different from \(a\), $c(a',x)$ can be anything. Similarly, if $f(x)$ is undefined, $c(a',x)$ can be arbitrary.
The \kl{least conditional} $c_0$ is undefined
in all such cases and contains no `junk'.
\begin{align*}
	c_0(a,x) = \begin{cases}
		b &\text{ if } f(x) \text{ is defined and } f(x) = (a,b) \\
		\text{undefined} & \text{ otherwise}
	\end{cases}.	
\end{align*}

\begin{proposition}
  Any \kl{cartesian bicategory of relations} has $\normalisedby$-least
  \kl{conditionals}.
\end{proposition}
\begin{proof}
  By Proposition~\ref{lemma:rel-partial-markov}, the morphism $c_{0} = (\id{} ⊗
  f) ⨾ ((μ ⨾ ε) ⊗ \id{})$ is a \kl{quasi-total} \kl{conditional} of 
  $f ፡ X
  → A ⊗ B$. Suppose there is another \kl{quasi-total} 
  \kl{conditional}, $c
  ፡ A ⊗ X → B$, of $f$.
  \[\includegraphics{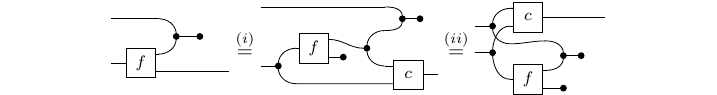}\]
  We use that (\emph{i}) \(c\) is a conditional of \(f\), and (\emph{ii}) the
  Frobenius axioms of \kl{cartesian bicategories of relations} to show that \(c_{0}
  \normalisedby c\).
\end{proof}

\section{Updates increase validity}%
\label{sec:validity_increase}%

In its most general form, bayesian updating involves a prior 
distribution and a (fuzzy) evidence predicate. From these inputs, one 
can compute a posterior distribution by incorporating the evidence 
into the prior. 

Given a distribution $σ$ and a predicate $p$, the validity of $p$ in $σ$ is a
real number in the unit interval $[0,1]$. This validity expresses the
probability that the predicate $p$ holds for a random sample from $σ$.

A key property of bayesian updating is that the validity of the 
evidence is
greater or equal in the posterior than in the prior. A proof of this fact for
discrete probability can be found in~\cite[Theorem~2.1]{jacobs2019learning}.
Ultimately, the result is a consequence of the well-known Cauchy--Schwarz
inequality. The \emph{Cauchy--Schwarz inequality} states that the square of the
inner product of two vectors $u,v \in \Real^n$ is smaller than the product of
their squared norms:
\[
\left(\sum\nolimits_{i} u_{i} · v_{i}\right)^{2} ≤
\left(\sum\nolimits_{i} u_{i}^{2}\right)
\cdot \left(\sum\nolimits_{j} v_{j}^{2}\right).
\]
The inequality also holds if we allow parameters $x$, $y$ and 
\(z\), thus considering families of vectors $h : X \to \Real^n$ $f : 
Y \to \Real^n$, and $g : Z \to \Real^n$. If we write $f_i(x)$ for 
$f(x)_i$ (similarly for $g$ and $h$) and let \(u_{i} = 
\sqrt{h_{i}(x)} \cdot f_{i}(y)\) and \(v_{j} = \sqrt{h_{j}(x)} \cdot 
g_{j}(z)\), the 
Cauchy--Schwarz inequality becomes:
\begin{equation}
	\label{eq:cs_params}
	\left(\sum\nolimits_{i}h_{i}(x) \cdot f_{i}(y) \cdot 
	g_{i}(z)\right)^{2} \leq \left(\sum\nolimits_{i} h_{i}(x) \cdot 
	f_{i}(y)^{2}\right) \cdot \left(\sum\nolimits_{j} h_{j}(x) \cdot 
	g_{j}(z)^{2}\right).
\end{equation}

We consider this general form and express it synthetically, using the 
\kl{conditional preorder}.

\begin{definition}%
  \textbf{(Cauchy--Schwarz inequality)}%
  \AP A \kl{discrete partial Markov category} satisfies the
  \intro{Cauchy--Schwarz inequality} if, for all triples of morphisms $h ፡ X →
  A$, $f ፡ A → Y$, and $g ፡ A → Z$, the following inequality holds.
  \[\includegraphics{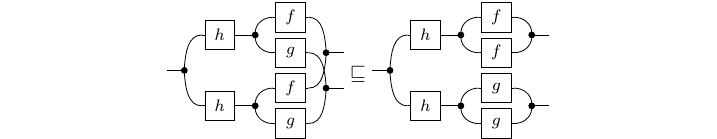}\]
\end{definition}

\begin{example}
  Instantiating the above definition in \(\FinSubStoch\) yields precisely the
  usual \emph{Cauchy--Schwarz inequality} (\ref{eq:cs_params}). Therefore,
  \(\FinSubStoch\) satisfies the synthetic \kl{Cauchy--Schwarz inequality}.
\end{example}

\begin{example}
  The \kl{Cauchy--Schwarz inequality} holds as equality in 
  discrete \kl{cartesian restriction categories}.
\end{example}

\begin{example}
  \kl{Cartesian bicategories of relations} satisfy the
  \kl{Cauchy--Schwarz
  inequality}.
  \[\cartesianbicategoriescauchyschwarzFig\]
\end{example}

\begin{proposition}
  A \kl{partial Markov category} that satisfies the \kl{Cauchy--Schwarz
  inequality} also satisfies the inequality below, called the \intro{means
  inequality}, for all $u ፡ X → A$ and $v ፡ A → Y$.
  \[\meansinequalityFig{}\]
\end{proposition}
\begin{proof}
  Let \(h = u\), \(f = v\) and \(g = \discard\). We apply \emph{(i)} the
  \kl{Cauchy--Schwarz inequality} and \emph{(ii)}
  Proposition~\ref{prop:subunitality}.
  \[\cauchyschwarzmeansinequalitiesProofFig{}\]
\end{proof}

\begin{remark}\textbf{(QM-AM inequality)}\label{prop:qmam:means}
    The means inequality owes its name to the inequality between the 
   weighted arithmetic and quadratic means of a vector of positive 
   real numbers, $x₁,\dots,xₙ$
  with weights $w₁ + \dots + wₙ = 1$.
  $$ \left( \sum\nolimits_{i ∈ I} w_i · x_i  \right)^2 ≤ \sum\nolimits_{i ∈ I} w_i · x_i^2.$$
\end{remark}

\begin{example}
  The \kl{partial Markov category} of substochastic channels, 
  \(\FinSubStoch\), satisfies the \kl{means inequality}.
  In particular, if $f$ is a state $\sigma : I \to X$, and $g$ is a predicate $p : X \to I$, then the means inequality gives the well-known inequality between expectations of random variables, usually written as $E[X]^2 \le E[X^2]$.
\end{example}

With this setup, we can work towards a synthetic analogue of 
the validity-increase theorem. We start by providing the relevant 
definitions.

\begin{definition}\textbf{(Validity)}\label{def:validity}\defining{linkValidity}
  In a \kl{partial Markov category}, the \emph{validity} of a predicate $p ፡ X → I$
  with respect to a prior $σ ፡ I → X$ is the scalar 
  obtained by composition,
  \((\sigma ⨾ p) : I \to I\). 
\end{definition}

\noindent
We first prove a special case of the validity increase theorem for deterministic evidence.

\begin{proposition}
	In a \kl{partial Markov category}, updating a prior state $σ : I → X$ with a
	 deterministic predicate, $p : X \to I$, increases the predicate's posterior
	 validity. That is,  
	\[ σ ⨾ p \leqprob p^\dagger_σ ⨾ p. \]
\end{proposition}
\begin{proof}
  The \kl{conditional} equation for the \kl{bayesian inverse} is $σ ⨾ 
  \copier_X
  ⨾ (p ⊗ \id_X) = (σ ⨾ p) ⊗ p^\dagger_σ$. Therefore, $σ ⨾ \copier_X ⨾ (p ⊗
  \id_X) \leqprob p^\dagger_σ$; and, as a consequence, \(σ ⨾ p = σ ⨾ \copier_X ⨾ (p ⊗ p) \leqprob
  p^\dagger_σ ⨾ p\).
  \[\validityincreasedeterministicProofFig{}\]
\end{proof}

When the predicate $p$ is not necessarily \kl{deterministic}, we will 
need the \kl{means inequality} and the following notion of zero 
scalars.	

\begin{definition}
  Let $\cat{C}$ be a \kl{copy-discard category}. 
  \begin{enumerate}
		\item A scalar $s : I \to I$ is called a \intro{zero scalar}, if it
		satisfies $s ⊗ f  = s ⊗ g$ for all $f,g : X → Y$. Any scalar that does
		not satisfy this property is said to be \AP\intro{non-zero}.
		
		\item We say that the non-zero scalars in $\cat{C}$ are
		\AP\intro{cancellative} if $s ⊗ f = s ⊗ g$ implies $f = g$ for all $f,g
		: X → Y$ and non-zero scalar $s : I → I$.
  \end{enumerate}
\end{definition}

\begin{example}
  We show that non-zero scalars in $\BorelSubStoch$ are cancellative.
	Scalars in $\BorelSubStoch$ correspond to elements of the unit interval $[0,1]$.
  Given a subprobability
	kernel $f : (X,\sigma_X) \to (Y,\sigma_Y)$, 
	the tensor product $s \otimes f$ is given by
	pointwise multiplication with $s \in [0,1]$. That is for $x \in 
	X$ and $A \in \sigma_Y$,
	\[ (s \otimes f)(A\given x) = s \cdot f(A \given x). \]
	Clearly, if $s \ne 0$, then $s \otimes f = s \otimes g$ implies 
	$f = g$ for parallel subprobability kernels $f$ and $g$. 
	Similarly, non-zero scalars in $\FinSubStoch$ are also 
	cancellative.
\end{example}

\begin{lemma}
  Let $\cat{C}$ be a \kl{copy-discard category} whose \kl{non-zero} scalars are
  \kl{cancellative}. Then, the following cancellation property holds for the
  "conditional inequality". If $s : I → I$ is a non-zero scalar and $f,g : X
  → Y$, then $s ⊗ f \leqprob s ⊗ g$ implies $f \leqprob g$.
\end{lemma}
\begin{proof}
	Let $s ⊗ f \leqprob_r s ⊗ g$. Then, by definition $f ⊗ s = (s ⊗ g) \condcomp
	r = s ⊗ (g \condcomp r)$. By cancellativity, $f = g \condcomp r$, and thus
	$f \leqprob_r g$.
\end{proof}

\begin{theorem}
  \label{thm:validity_increase}
  Let $\cat{C}$ be a \kl{partial Markov category} that satisfies the 
  \kl{means inequality} and whose "non-zero" scalars are "cancellative".
  Then, updating a prior state with a predicate increases the 
  predicate's validity in the posterior. That is, for all $σ : I 
  → X$ and $p : X → I$,
  \[σ ⨾ p \leqprob p^\dagger_{σ} ⨾ p.\]
\end{theorem}
\begin{proof}
  We apply the \kl{means inequality} in the following calculation.
  \[\includegraphics{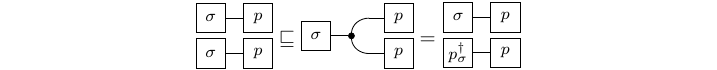}\]
  Since the non-zero scalars are cancellative, we conclude that \(σ ⨾ p \leqprob
  \bayesinv{p}{σ} ⨾ p\), because the scalar is either zero or cancellative.
\end{proof}

Notably, $\FinSubStoch$ satisfies the assumptions of the previous theorem.
Therefore, we have synthetically recovered Jacobs' ``updates increase validity''
\cite[Theorem 2.1]{jacobs2019learning}. Also observe that the proof of Theorem
\ref{thm:validity_increase} does not rely on the full power of the means
inequality. The instance we used in the proof of Theorem
\ref{thm:validity_increase} holds in $\BorelSubStoch$, as the following example
demonstrates.
\begin{example}
	Let $(X,\sigma_X)$ be standard Borel, and $\mu$ a subprobability 
	measure on $X$. The collection of $\mu$-integrable functions $X 
	\to \Real$, quotiented by $\mu$-almost-sure equality, forms a real
	inner product space with inner product $\langle 
	p,q\rangle = \int p(x) \cdot q(x) \cdot \mu(dx)$. Therefore, one 
	can apply the Cauchy--Schwarz inequality to get
	\[ \Bigl(\int p(x) \cdot \mu(dx)\Bigr)^2 = \langle 
	p,\discard_{X}\rangle^2 \le \langle p,p\rangle \cdot \langle 
	\discard_{X},\discard_{X} \rangle = \Bigl(\int p(x)^2 \cdot 
	\mu(dx)\Bigr) \cdot 
	\Bigl( \int 1 \cdot \mu(dx) 
	\Bigr) \le \int p(x)^2 \cdot 
	\mu(dx). \]
\end{example}

\section{Concluding remarks}%
\label{sec:concludingRemarks}%

We introduced a novel preorder enrichment in the general setting of
\partialMarkovCategories{}; we showed that it generalizes the canonical preorder
enrichments of two categories of probability kernels, \kl{cartesian restriction
categories}, and \kl{cartesian bicategories of relations}. We showed how
comparators relate to least conditionals of copier maps. We axiomatized the
\kl{Cauchy--Schwarz inequality} and the quadratic-arithmetic \kl{means
inequality}. From these, we abstractly derived the fact that Bayesian updating
increases the validity of the evidence in the posterior.

Further exploration of the basic theory of the \conditionalInequality{} is
warranted. One interesting question is about antisymmetry of the order, which is
related to inverses of scalars. An inverse of $s : I → I$ is a scalar $s^{-1} :
I → I$ such that $s ⊗ s^{-1} = \discard_I$. Clearly, \conditionalInequality{} is
antisymmetric only if the only scalar with an inverse is 
$\discard_I$. It is an open question whether the converse statement 
holds. It would also be
interesting to find necessary conditions for the "conditional preorder" to form
an enrichment. Relatedly, the synthetic formulation of the \kl{Cauchy--Schwarz
inequality} and the means inequality opens the way for a more refined
axiomatization: can we derive both from some more fundamental principle?

Finally, a potential application of the synthetic \kl{means inequality} lies in
mixing it with the additive structure of some \kl{partial Markov categories},
like effectuses with copiers \cite{Cho2015}. In that setting, it may be possible
to develop a synthetic theory of variance and covariance. 

\subsection*{Acknowledgements}

We want to thank \textsf{Dario Stein} for providing the example of standard
Borel spaces (\S \ref{ex:standardBorel}). We thank the anonymous reviewers for
much constructive and detailed feedback; we thank an anonymous reviewer for
pointing out the recent work on quasi-Markov categories.

\textit{Funding:} Elena Di Lavore, Mario Román, and Paweł Sobociński were
supported by the Advanced Research + Invention Agency (ARIA) Safeguarded AI
Programme. Paweł Sobociński was additionally supported by European Union and
Estonian Research Council via grants PRG1215 and TEM-TA5 and the Estonian Center
of Excellence in Artificial Intelligence.  Mario Román was additionally
supported by the Air Force Office of Scientific Research (AFOSR) award number
FA9550-21-1-0038. This article is based upon work from COST Action EuroProofNet,
CA20111, supported by COST (European Cooperation in Science and Technology).

\bibliographystyle{entcs}
\bibliography{main}
\end{document}